\documentclass[10pt, twocolumn, Journal]{IEEEtran}

\newcommand{\subparagraph}{}
\usepackage[compact]{titlesec}
\titlespacing*{\section}{8pt}{1.1\baselineskip}{1.0\baselineskip}

\titlespacing\section{0pt}{8pt plus 4pt minus 2pt}{5pt plus 2pt minus 2pt}
\titlespacing\subsection{0pt}{8pt plus 4pt minus 2pt}{3pt plus 2pt minus 2pt}
\titlespacing\subsubsection{0pt}{6pt plus 4pt minus 2pt}{3pt plus 2pt minus 2pt}

\usepackage[numbers,sort&compress]{natbib}

\usepackage{amsmath, mathrsfs, mathtools}
\usepackage{amsfonts}
\usepackage{amssymb}

\usepackage[pdftex]{graphicx}
\usepackage{caption}
\usepackage{subfig, float}

\usepackage{color}
\usepackage{multirow}
\usepackage{stmaryrd}

\allowdisplaybreaks

\usepackage{amsfonts}
\usepackage{amssymb}
\usepackage{amsthm}
\usepackage{color}
\usepackage{multirow}

\usepackage{tikz}
\usepackage{pgfplots,tikz-3dplot}
\pgfplotsset{compat=newest}
\usetikzlibrary{patterns}
\usetikzlibrary{positioning}
\usetikzlibrary{datavisualization}
\usetikzlibrary{datavisualization.formats.functions}
\usetikzlibrary{backgrounds}
\usetikzlibrary{shapes,snakes}

\long\def\comment#1{}

\def\Gammam{{\boldsymbol{\Gamma}}}

\newcommand{\Sigmam}{\boldsymbol{\Sigma}}

\newcommand{\Psim}{\boldsymbol{\Psi}}

\renewcommand{\det}{{\hbox{det}}}

\newcommand{\transp}{{\sf T}}

\usepackage{mathbbol}

\def\mindex#1{\index{#1}}

\def\sq{\hbox{\rlap{$\sqcap$}$\sqcup$}}
\def\qed{\ifmmode\sq\else{\unskip\nobreak\hfil
\penalty50\hskip1em\null\nobreak\hfil\sq
\parfillskip=0pt\finalhyphendemerits=0\endgraf}\fi\medskip}

\long\def\defbox#1{\framebox[.9\hsize][c]{\parbox{.85\hsize}{%
\parindent=0pt
\baselineskip=12pt plus .1pt      
\parskip=6pt plus 1.5pt minus 1pt 
 #1}}}

\long\def\beginbox#1\endbox{\subsection*{}%
\hbox{\hspace{.05\hsize}\defbox{\medskip#1\bigskip}}%
\subsection*{}}

\def\endbox{}

\def\diag{{\text{diag}}}

\def\tr{\mathsf{tr}}

\newsavebox{\junk}
\savebox{\junk}[1.6mm]{\hbox{$|\!|\!|$}}

\def\det{{\mathop{\rm det}}}

\def\argmin{\mathop{\rm arg\, min}}

\def\bC{{\mathbb C}}
\def\bD{{\mathbb D}}
\def\bE{{\mathbb E}}

\def\bR{{\mathbb R}}

\def\bfA{{\bf A}}

\def\bfC{{\bf C}}
\def\bfD{{\bf D}}

\def\bfI{{\bf I}}

\def\bfK{{\bf K}}
\def\bfL{{\bf L}}

\def\bfV{{\bf V}}

\def\bfX{{\bf X}}

\def\bfa{{\bf a}}

\def\bfc{{\bf c}}

\def\bfe{{\bf e}}

\def\bfh{{\bf h}}

\def\bfu{{\bf u}}
\def\bfv{{\bf v}}
\def\bfw{{\bf w}}
\def\bfx{{\bf x}}

\def\bfz{{\bf z}}

\def\ttE{{\mathtt{E}}}

\def\ttM{{\mathtt{M}}}

\def\ttS{{\mathtt{S}}}

\def\ttV{{\mathtt{V}}}

\def\sfF{{\sf F}}

\def\bfmath#1{{\mathchoice{\mbox{\boldmath$#1$}}%
{\mbox{\boldmath$#1$}}%
{\mbox{\boldmath$\scriptstyle#1$}}%
{\mbox{\boldmath$\scriptscriptstyle#1$}}}}

\def\bfmY{\bfmath{Y}}

\def\bfmhhaY{\bfmath{\hhaY}} 
\def\bfmhhaY{\hbox to 0pt{$\widehat{\bfmY}$\hss}\widehat{\phantom{\raise 1.25pt\hbox{$\bfmY$}}}}

\def\til={{\widetilde =}}

\def\clA{{\cal A}}

\def\clG{{\cal G}}
\def\clH{{\cal H}}

\def\clK{{\cal K}}

\def\clX{{\cal X}}

 \def\FRAC#1#2#3{\genfrac{}{}{}{#1}{#2}{#3}}

\def\ddtp{{\mathchoice{\FRAC{1}{d^{\hbox to 2pt{\rm\tiny +\hss}}}{dt}}%
{\FRAC{1}{d^{\hbox to 2pt{\rm\tiny +\hss}}}{dt}}%
{\FRAC{3}{d^{\hbox to 2pt{\rm\tiny +\hss}}}{dt}}%
{\FRAC{3}{d^{\hbox to 2pt{\rm\tiny +\hss}}}{dt}}}}

\def\average#1,#2,{{1\over #2} \sum_{#1}^{#2}}

\def\eye(#1){{\bf(#1)}\quad}

\newtheorem{theorem}{{\bf Theorem}}

\def\eq#1/{(\ref{e:#1})}

\newcommand{\beqn}[1]{\notes{#1}%
\begin{eqnarray} \elabel{#1}}

\newcommand{\eeqn}{\end{eqnarray} }

\newcommand{\beq}[1]{\notes{#1}%
\begin{equation}\elabel{#1}}

\newcommand{\eeq}{\end{equation}}

\def\bdes{\begin{description}}
\def\edes{\end{description}}

\newcounter{rmnum}

\newcounter{anum}

{\end{list}}

\def\ass(#1:#2){(#1\ref{#1:#2})}

\def\ritem#1{
\item[{\sf \ass(\current_model:#1)}]
}

\newenvironment{recall-ass}[1]{%
\begin{description}
\def\current_model{#1}}{
\end{description}
}

\def\bfMse{{\ttM\ttM\ttS\ttE}}

\def\gammam{\boldsymbol{\gamma}}

\usepackage{afterpage}
\usepackage{balance}

\def\herm{{\sfT}}
\def\herm{{\dagger}}

\newcommand{\normd}[1]{{\left\vert\kern-0.25ex\left\vert\kern-0.25ex\left\vert #1 
    \right\vert\kern-0.25ex\right\vert\kern-0.25ex\right\vert}}

\def\argmin{\mathop{\rm arg\, min}}
\def\alg{{\ttM\ttM\ttV}}

\title{Multiple Measurement Vectors Problem: A Decoupling Property and its Applications}
\author{Saeid Haghighatshoar,  \IEEEmembership{Member, IEEE,} Giuseppe Caire,
	\IEEEmembership{Fellow, IEEE}%
\\
\vspace{1mm}

Emails: saeid.haghighatshoar@tu-berlin.de, caire@tu-berlin.de
	\thanks{The authors are with the Communications and Information Theory Group, Technische Universit\"{a}t Berlin (\{saeid.haghighatshoar, caire\}@tu-berlin.de).}
	\vspace{-4mm}
}

\begin{document}

\maketitle

\begin{abstract}
We  study a Compressed Sensing (CS) problem known as \textit{Multiple Measurement Vectors} (MMV) problem,  which arises in joint estimation of multiple signal realizations  when the signal samples have a common (joint) sparse support over a fixed known dictionary. Although there is a vast literature on the analysis of MMV, it is not yet fully known how the number of signal samples and their statistical correlations affects the performance of the joint estimation in MMV.
%
 Moreover,  in many instances  of MMV the underlying sparsifying dictionary may not be precisely known, and it is still an open problem to quantify how the dictionary mismatch may affect the estimation performance. 
 
 In this paper, we focus on \textit{$\ell_{2,1}$-norm regularized least squares} ($\ell_{2,1}$-LS) as a well-known and widely-used  MMV algorithm in the literature. We prove an interesting decoupling property for $\ell_{2,1}$-LS, where we show that it can be decomposed into two phases:   i)\,use all the signal samples to estimate the signal covariance  matrix (coupled phase), ii)\,plug in the resulting covariance estimate as the true covariance matrix into the \textit{Minimum Mean Squared Error} (MMSE) estimator to reconstruct each signal sample \textit{individually} (decoupled phase). As a consequence of this decomposition, we are able to provide further insights on the performance of $\ell_{2,1}$-LS for MMV. In particular, we address how the signal correlations and dictionary mismatch affects its   performance.
 Moreover, we show that by using the decoupling property one can obtain a variety of MMV algorithms with  performances even better than that of $\ell_{2,1}$-LS.
  We also provide numerical simulations to validate our theoretical results.  

\end{abstract}
\begin{keywords}
Compressed Sensing,  Sparsifying Dictionary, Multiple Measurement Vectors, Mean Squared Error. 
\end{keywords}

\vspace{-2mm}

\section{Introduction}
Efficient and reliable signal estimation in a variety of fields including biology, statistics, wireless communication, etc.  requires adopting sparsity in a suitable signal dictionary as signal prior and using techniques and recovery algorithms from Compressed Sensing (CS) \cite{donoho2006compressed, candes2006near}. To streamline the explanation of the main ideas in this paper, we start with and keep as an illustrative example throughout the paper a signal estimation problem encountered in wireless  MIMO systems, although as we will generalize later, the underlying signal model and our  results apply also to other similar signal processing problems. 

\noindent{\bf Illustrative Example.} We consider a wireless propagation model illustrated in Fig.\,\ref{array_fig}, where a BS is equipped with a \textit{Uniform Linear Array} (ULA) consisting of $n\gg 1$ antennas and serves single-antenna users. 
\begin{figure}[t]
	\centering
	\includegraphics[width=0.3\textwidth]{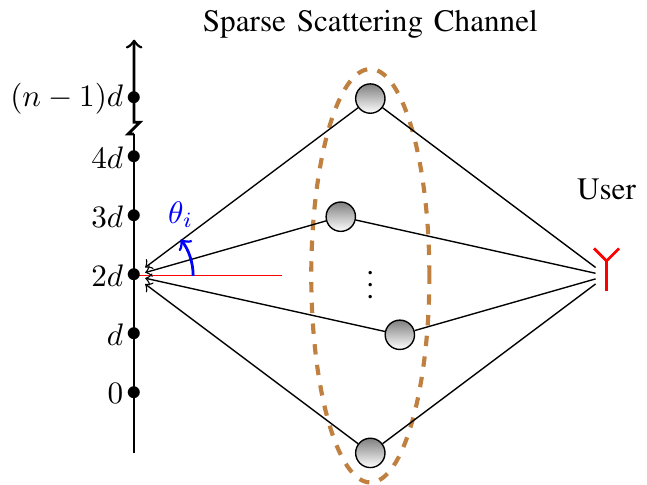}
	\caption{Illustration of a sparse signal model at a BS receiver with $n$ antennas produced by a scattering channel consisting of a sparse set of scatterers in the angular domain.  }
	\label{array_fig}
\end{figure}
Here, we assume that there are $T\gg 1$ resource blocks (signal samples) and the signal from a generic user received at $n$ BS antennas at resource block $s \in [T]:=\{1,\dots, T\}$ is given by \cite{tse2005fundamentals, haghighatshoar2017massive} 
\begin{align}\label{chan_vec1}
\bfh(s)= \sum_{i=1}^k w_i(s) \bfa(\xi_i),
\end{align}
where $w_i(s)$ denotes the channel gain of the $i$-th scatter, $i \in [k]$, at resource block $s$, where $\xi_i =\sin(\theta_i) \in \Xi:=[-1,1]$ parametrizes the  \textit{angle-of-arrival} (AoA) $\theta_i \in [-\pi, \pi]$ of the $i$-th scatterer, and where $\bfa(\xi) \in \bC^n$ denotes the array response to a planar wave coming from the AoA $\xi$ given by
\begin{align}\label{array_response}
\bfa(\xi)=(e^{j \pi \xi}, \dots, e^{j n \pi \xi})^\transp.
\end{align}
We consider a quasi-stationary scenario, as in most or typical MIMO and array processing problems relevant for communication\footnote{There is a vast literature on wideband array processing and also non-stationary scenarios where this assumption is not met, but it typically pertains to radar rather than communication.}, where we assume that the AoAs $\{\xi_i: i \in [k]\} \subset \Xi$ in signal model \eqref{chan_vec1} remain the same across all the signal samples $\clH:=\{\bfh(s): s\in [T]\}$ such that they all have the same support in the AoA domain $\Xi$, whereas the channel gains $w_i(s)$,  might vary considerably (even i.i.d.) across  different scatterers $i\in[k]$ and signal samples $s\in [T]$.

\noindent{\bf Joint Sparsity Model.} For the rest of the paper, we will use the signal model \eqref{chan_vec1} and will consider a stochastic  process $\clH:=\{\bfh(s): s\in [T]\}$ whose samples have a sparse representation over a common dictionary given by
\begin{align}\label{A_dict}
\clA:=\{\bfa(\xi): \xi \in \Xi\},
\end{align}
where the  atoms of the dictionary $\bfa(\xi)\in \bC^n$ are labelled with a finite or infinite set $\Xi$ (e.g., $\Xi= [-1,1]$ for the continuum set of AoAs). In particular, all signal samples share a common sparsity pattern or  support $\{\xi_i: i \in [k]\}$  over $\Xi$ of size $k \ll n$, where $n$ is the signal dimension. 


\subsection{Single and Multiple Measurement Vector problem} \label{sec:S-MMV}
In many signal processing problems, one needs to recover the signal $\bfh(s)$ as in \eqref{chan_vec1} from a set of $m \ll n$ linear and possibly noisy projections $\bfx(s)=\Psim(s) \bfh(s) + \bfz(s)$, where $\Psim(s) \in \bC^{m \times n}$ is the $m \times n$ projection matrix applied to the sample $s$ and where $\bfz(s)$ is the measurement noise. Low-dim projections ($m \ll n$) arise in signal processing problems due to having less (here $m$) measurement sensors than the signal dimension $n$.\footnote{For example, in our illustrative example in Fig.\,\ref{array_fig}, this happens when the receiver has  only $m \ll n$ antenna ports (RF demodulation chains and A/D converters) such that the $n$ antennas are connected with the $m$ antenna ports through a linear network of phase shifters, attenuators, and signal adders, which can be represented by an $m \times n$ matrix  $\Psim$ in the discrete-time complex baseband  representation \cite{rimoldi2016principles} of the signal obtained after A/D conversion.}
 With sparsity over the dictionary $\clA$ in \eqref{A_dict} as the signal prior, a well-known and widely used approach for signal reconstruction is to apply the $\ell_1$-norm regularized least squares ($\ell_1$-LS) minimization (generally known as LASSO in signal processing and statistics literature \cite{tibshirani1996regression}) given by
\begin{align}\label{l1_ls}
\widehat{\bfc}(s)=\argmin_{\bfc \in \bC^G} \frac{1}{2} \|\bfx(s) - \Psim(s) \bfA \bfc\|_2^2 + \varrho  \|\bfc\|_1,
\end{align}
where $\bfA=[\bfa(\xi^\circ_1)\,, \dots, \bfa(\xi^\circ_G)]$ is a finite dictionary matrix generally obtained by an approximation (finite dictionary) or a quantization (infinite dictionary) of the sparsifying dictionary $\clA$ in \eqref{A_dict} over a finite grid $\clG=\{\xi^\circ_1,  \dots, \xi^\circ_G\}$ of size $G$, where $\bfc=(c_1, \dots, c_G)^\transp  \in \bC^G$ denote the coefficients of the approximation of $\bfh(s)$ over the discrete dictionary $\bfA$, where $\|\bfc\|_1=\sum_{i\in [G]} |c_i|$ denotes the $\ell_1$-norm of $\bfc$, and where $\varrho>0$ is a regularization parameter. 
It is well-known in CS \cite{chen2001atomic, chandrasekaran2012convex} that the $\ell_1$-norm of $\bfc$ in the regularization term can be seen as a convex relaxation of the $\ell_0$-norm of $\bfc$, defined  by $\|\bfc\|_0=\big | \{i \in [G]: c_i \not = 0\} \big |$ as the number of nonzero elements of $\bfc$, and naturally promotes the sparsity of the coefficients in $\bfc$. As a result, in the ideal scenario where the grid $\clG$ contains the true signal support $\{\xi_1, \dots, \xi_k\} \subset \Xi$ (see, e.g., \eqref{chan_vec1}), we expect that the coefficient vector $\widehat{\bfc}(s)$ has non-zero values on the elements corresponding to the true atoms and is zero elsewhere. Hence, one can estimate the signal $\bfh(s)$ from the estimated sparse coefficients $\widehat{\bfc}(s)$ as $\widehat{\bfh}(s)=\bfA \widehat{\bfc}(s)$. 

One can run the optimization problem \eqref{l1_ls} to estimate each signal sample $\bfh(s)$ from its  corresponding noisy sketches $\bfx(s)=\Psim(s) \bfh(s) + \bfz(s)$ individually; this is known as the \textit{Single Measurement Vector} (SMV) problem in the literature. In the presence of joint sparsity, however, it is convenient to run a joint CS estimation problem to improve the recovery performance. This is known as the \textit{Multiple Measurement Vectors} (MMV) problem and is quite well studied in the literature \cite{tropp2006algorithms, tropp2006algorithms2, malioutov2005sparse, lee2012subspace, davies2012rank}. A well-known algorithm and an MMV variant of $\ell_1$-LS in \eqref{l1_ls} is $\ell_{2,1}$-norm regularized least squares ($\ell_{2,1}$-LS) 
\begin{align}\label{l21_ls}
\widehat{\bfC}=\argmin_{\bfC \in \bC^{G\times T}} \frac{1}{2} \sum_{s \in [T]} \|\bfx(s) - \Psim(s) \bfA \bfc(s)\|_2^2 + \varrho  \sqrt{T} \|\bfC\|_{2,1},
\end{align}
where $\widehat{\bfC}=\big [ \widehat{\bfc}(1),  \dots, \widehat{\bfc}(T) \big ]$ is the $G\times T$ matrix of estimated coefficients with $\widehat{\bfc}(s)$ denoting the estimated coefficient vector of the $s$-th signal $\bfh(s)$, and where $\|\bfC\|_{2,1}=\sum_{i\in [G]} \|\bfC_{i,:}\|_2$ is the sum of $\ell_2$-norm of rows of $\bfC$ with $\bfC_{i,:}$ denoting the $i$-th row of $\bfC$, and where $\varrho >0$ is a regularization parameter as before. Similar to $\ell_1$-LS, one can  estimate  the original signal samples as $\widehat{\bfh}(s)=\bfA \widehat{\bfc}(s)$, $s \in [T]$.
 It is known that the $\ell_{2,1}$-norm regularization promotes the row sparsity of the coefficient matrix $\bfC$, thus, it imposes a common (joint) sparsity pattern on the supports of the estimated signals, which is the desirable regularization to adopt in the MMV setups where the signal samples are jointly sparse.
\subsection{Related Works and Contribution}\label{sec:rel_work}
There is a vast literature in CS on MMV and signal estimation under joint/group sparsity; we may refer to \cite{tropp2006algorithms, tropp2006algorithms2, malioutov2005sparse, lee2012subspace, davies2012rank} for several earlier works and the refs. therein. A broad overview of the algorithms proposed for signal recovery in the MMV setup reveals that they generally belong to the following  two classes of algorithms. 
The first class applies greedy techniques or 
convex optimization using joint sparsity promoting regularizers  (e.g., $\ell_{2,1}$-norm regularization as in $\ell_{2,1}$-LS in \eqref{l21_ls}) to estimate the signals. The second class of algorithms use the sample covariance matrix of data and applies subspace techniques to extract the joint support. 
Once the support  is identified, the standard least squares method can be applied to find the corresponding coefficients. 
Our results in this paper, broadly speaking,  make a link between these two class of algorithms. More precisely, by deriving a \textit{decoupling property} in the specific case of $\ell_{2,1}$-LS belonging to the first class of algorithms, we show that $\ell_{2,1}$-LS can be decomposed into two steps where the first step can be seen merely as the estimation of the sparse covariance of the signal samples, which highly resembles the subspace techniques  applied in the second class of algorithms. 

All the classical algorithms for MMV require a finite signal dictionary, thus, when the signal dictionary is infinite (as in the continuum AoA in our illustrative example), one needs to run MMV on a finite quantized version of the original dictionary. This causes a mismatch when the signal support is not covered by the quantized dictionary and may dramatically degrade the performance \cite{chi2011sensitivity}. Recently off-grid MMV techniques have also been developed  
\cite{tang2013compressed, li2016off, yang2014exact} to avoid the dictionary mismatch.
Using the decoupling property for $\ell_{2,1}$-LS, we are able to fully quantify the effect of dictionary mismatch in terms of the convex cone produced by the dictionary, which provides further insights on the effect of dictionary quantization.

Recently a novel class of algorithms based on \text{Vector Approximate Message Passing} (VAMP) have been proposed for signal recovery in MMV. VAMP belongs to the Bayesian class of algorithms and requires knowing the probability distribution of the coefficients of signal samples (e.g., $\{w_i(s): s \in [T]\}$ in \eqref{chan_vec1} in our illustrative example). The  performance of VAMP for i.i.d. Gaussian dictionaries can be fully specified in the  large-dimensional setup by  a simple \textit{State Evolution} equation \cite{kim2011belief}, which yields an exact characterization of the performance in terms of parameters such as signal distribution and number of MMV samples $T$. In practice, however, signal dictionaries are structured and, in particular, far from i.i.d. Gaussian, and exact analysis of VAMP in those cases is still an open problem. The decoupling property for $\ell_{2,1}$-LS reveals that the  performance of MMV over any dictionary (even a structured one) hinges crucially on the quality of the covariance estimation phase. Although, in this paper we do not make further progress in this direction, we provide new guidelines on how one may be able to theoretically analyze the  performance of  MMV for such structured dictionaries.

 \subsection{Notation}
 We have already introduced some of the important notations. We use non-boldface letters for scalars (e.g., $x$), boldface letters for vectors (e.g. $\bfx$), and capital boldface letters for matrices (e.g., $\bfX$). We denote the $i$-th row and the $i$-th column of a matrix $\bfX$ with $\bfX_{i,:}$ and $\bfX_{:,i}$. The identity matrix of order $m$ is denoted by $\bfI_{m}$. We denotes a diagonal $G \times G$ matrix with diagonal elements $d_1, \dots, d_G$ by $\diag(d_1,\dots, d_G)$. We use calligraphic letters for sets (e.g., $\clX$) and denote their cardinality by $|.|$ (e.g., $|\clX|$). 
 
 \section{Overview of the Results}\label{sec:summary}
\subsection{Signal Samples as Side Information in MMV}
Let us consider the joint optimization \eqref{l21_ls}. As  explained before, we can estimate each signal sample $\bfh(s)$, $s\in[T]$, from the optimal solution $\widehat{\bfC}$ of \eqref{l21_ls} as $\widehat{\bfh}(s)=\bfA \widehat{\bfc}(s)$, where $\widehat{\bfc}(s)$ denotes the $s$-th column of $\widehat{\bfC}$. 
This can be interpreted as follows: $\ell_{2,1}$-LS  yields an MMV estimator $\widehat{\bfh}(s)=\alg(\bfx(s); \bfX_{:, \sim s})$ for each signal sample $s\in [T]$, which depends \textit{explicitly} on the corresponding sketch $\bfx(s)$ of signal sample $\bfh(s)$ as well as \textit{implicitly} on all the other sketches $\bfX_{:, \sim s}$, where $\bfX=[\bfx(1), \dots, \bfx(T)]$ denotes the $m \times T$ matrix of sketches, and where $\bfX_{:, \sim s}$ denotes the $m \times (T-1)$ matrix obtained after removing the $s$-th column of $\bfX$ corresponding to $\bfx(s)$. One can interpret $\bfX_{:, \sim s}$ as some sort of \textit{side information} from the other sketches, exploited by MMV in the estimation of $\bfh(s)$. Our goal in this paper is to better understand how this side information is used/processed by MMV. 

To explain this more clearly, let us consider as an alternative to MMV, a genie-aided \textit{Minimum Mean Square Error} (MMSE)\footnote{For simplicity, we focused on a Gaussian process $\clH$, for which the MMSE estimator is linear. The results remain valid for a non-Gaussian process by replacing the MMSE  by L-MMSE (linear MMSE), which depends only on the second-order statistics of the signal and takes on the same form for any signal statistics as far as the covariance of the target vector $\bfh(s)$ and cross-covariance between $\bfh(s)$ and its corresponding sketches $\bfx(s)$ remain the same \cite{kay1993fundamentals}.} 
estimator, denoted by $\bfMse$, that has a perfect knowledge of the covariance matrix of the samples $\Sigmam_\bfh=\bE[\bfh(s) \bfh(s)^\herm]$. 
In the conventional case, where the signal coefficients $\{w_i(s): i \in [k]\}$ are independent of each other, $\Sigmam_\bfh$ from \eqref{chan_vec1} is given by 
$\Sigmam_\bfh=\sum_{i=1}^k \gamma_i^\circ \bfa(\xi_i) \bfa(\xi_i) ^\herm$, where $\gamma_i^\circ=\bE[|w_i(s)|^2]$ denotes the strength of the $i$-th coefficient. Let us consider a scenario where the signal samples are i.i.d. Then, the estimate $\widehat{\bfh}(s)$ of such a genie-aided algorithm from all the available sketches $\bfX$ would be 
\begin{align}\label{mmse_intro}
\widehat{\bfh}(s)&=\bfMse(\bfh(s)|\bfX)\stackrel{(a)}=\bfMse(\bfh(s)| \bfx(s))\stackrel{(b)}{=}\Sigmam_{\bfh \bfx} \Sigmam_{\bfx}^{-1} \bfx(s)\nonumber\\
&=\Sigmam_\bfh \Psim(s)^\herm \Big (\Psim(s) \Sigmam_\bfh \Psim(s)^\herm + \sigma^2 \bfI_m \Big )^{-1} \bfx(s),
\end{align}
where $\Sigmam_{\bfx}=\bE[\bfx(s) \bfx(s)^\herm]$ and $\Sigmam_{\bfh \bfx}=\bE[\bfh(s) \bfx(s)^\herm]$ denote the covariance matrix of $\bfx(s)$ and the cross covariance matrix of $\bfh(s)$ and $\bfx(s)$ respectively (assuming i.i.d. noise samples of variance $\sigma^2$),  where $(a)$ follows from the independence of the signal samples, and where in $(b)$ we applied the well-known formula of MMSE estimator for Gaussian signals \cite{kay1993fundamentals}.

From \eqref{mmse_intro}, it is seen that, in contrast with the MMV which performs a joint estimation using all the sketches $\bfX$, the output $\widehat{\bfh}(s)$ of MMSE at each sample $s$  depends only on the sketch $\bfx(s)$ but not on the other sketches $\bfX_{:, \sim s}$. Therefore, we are curious to understand how the side information coming from the other samples is exploited by MMV and how it improves the estimation performance of an individual signal sample.

\subsection{Main Result: Decoupling Property of MMV}

\noindent{\bf MMV is a plug-in MMSE estimator:} For a fixed $s\in [T]$, let $\widehat{\bfh}(s)=\alg(\bfx(s); \bfX_{:, \sim s})$ be the MMV estimator of $\bfh(s)$ based on the input $\bfx(s)$ and the side information $\bfX_{:, \sim s}$. We prove that  MMV estimator can be decomposed into two parts: 
\begin{enumerate}
\item Use the sketch $\bfx(s)$ and all the side information $\bfX_{:, \sim s}$ to obtain an estimate $\widehat{\Sigmam}_{\bfh}$ of the covariance matrix of the samples $\clH=\{\bfh(s): s\in [T]\}$ (\textit{coupled phase}). 

\item Treat the resulting covariance estimate $\widehat{\Sigmam}_{\bfh}$ as the true covariance matrix $\Sigmam_\bfh$, plug it in the $\bfMse$ (see, e.g., \eqref{mmse_intro}), and find an estimate $\widehat{\bfh}(s)$ of each signal sample $\bfh(s)$ from its corresponding sketch $\bfx(s)$ \textit{individually} according to \eqref{mmse_intro} (\textit{decoupled phase}).
\end{enumerate}
We call this the \textit{decoupling property} of $\ell_{2,1}$-LS. 
Fig.\,\ref{fig:mmv_dec} illustrates this underlying coupled/decoupled phases for MMV. In the following, we will prove this property theoretically and discuss some of its crucial implications in  joint signal estimation through $\ell_{2,1}$-LS.

\begin{figure}[t]
\centering
	
	\subfloat[\label{dc_1}]{%
		\includegraphics[width=0.27\textwidth]{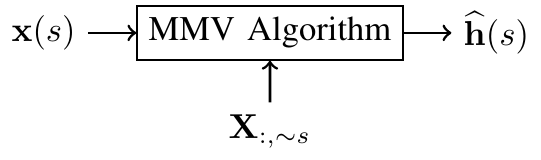}
        }
	
	\subfloat[\label{dc_2}]{%
		\includegraphics[width=0.45\textwidth]{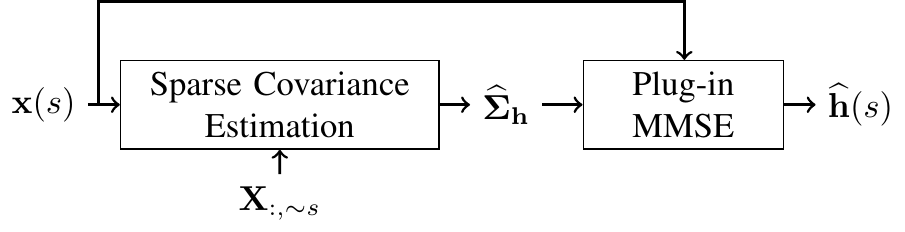}
	}
\caption{Illustration of MMV  from  perspective of a generic sample $\bfh(s)$: Fig.\,\ref{dc_1} treats the other samples $\bfX_{:, \sim s}$ as the side information to MMV, and Fig.\,\ref{dc_2} shows the underlying decomposition of MMV as a plug-in MMSE estimator. }
\label{fig:mmv_dec}
\end{figure}

\section{Main Results and Discussion}
\subsection{Main Theorems}
Let $\clH=\{\bfh(s): s\in [T]\}$ be the $n$-dim jointly sparse stochastic process as in signal model \eqref{chan_vec1} and let $\bfx(s)=\Psim(s) \bfh(s) + \bfz(s)$ be the collection of sketches of $\clH$ taken via $m \times n$ projection matrices $\Psim(s)$. 
For simplicity, we will assume that $\Psim(s) \Psim(s)^\dagger=\bfI_m$, where $\bfI_m$ denotes the identity matrix of order $m$.
Let us define $\ell_{2,1}$-LS cost function as in \eqref{l21_ls} by
\begin{align}\label{cost_mmv}
f(\bfC)=\frac{1}{2} \sum_{s \in [T]} \|\bfx(s) - \Psim(s) \bfA \bfc(s)\|_2^2 + \varrho  \sqrt{T} \|\bfC\|_{2,1},
\end{align}
where $\bfA$ denotes the $n\times G$ dictionary matrix and 
where $\bfC=[\bfc(1), \dots, \bfc(T)]$ denotes the $G \times T$ matrix of coefficients of $\clH$ over $\bfA$.

Let us also introduce the following convex cost function
\begin{align}\label{cost_gam}
g(\gammam)=& \frac{1}{T}\sum_{s\in [T]} \bfx(s)^\herm \big (\Psim(s)\bfA \Gammam \bfA^\herm \Psim(s)^\herm + \varrho \bfI_m \big )^{-1} \bfx(s)\nonumber\\
& \ \ + \tr(\Gammam),
\end{align}
where $\gammam=(\gamma_1, \dots, \gamma_G)^\herm \in \bR_+^G$ is a $G$-dim vector consisting of non-negative elements,  where $\Gammam:=\diag(\gammam)$ is a diagonal matrix with diagonal elements $\gammam$, and where $\tr(.)$ denotes the trace operator.
We first prove the following theorem which links the optimal solutions of the convex cost functions in \eqref{cost_mmv} and \eqref{cost_gam}.
\begin{theorem}\label{eq_thm}
Let $\widehat{\bfC}=\argmin_{\bfC \in \bC^{G \times T}} f(\bfC)$ and $\widehat{\gammam}=\argmin_{\gammam \in \bR_+^G} g(\gammam)$ be the optimal solutions of the convex functions $f(\bfC)$ and $g(\gammam)$ as in \eqref{cost_mmv} and \eqref{cost_gam}, respectively. Then, $\widehat{\gamma}_i=\frac{\|\widehat{\bfC}_{i,:}\|_2}{\sqrt{T}}$ for $i \in [G]$. \hfill$\square$
\end{theorem}
\begin{proof}
Proof in Appendix \ref{eq_thm_app}.
\end{proof}
Our second theorem makes a more direct link between the optimal solution $\gammam$ of  $g(\gammam)$ in \eqref{cost_gam} and the signal estimates produced by $\ell_{2,1}$-LS via optimizing $f(\bfC)$ in \eqref{cost_mmv}.
\begin{theorem}\label{mmse_eq_thm}
Let $\widehat{\bfC}$ and $\widehat{\gammam}$ be as in Theorem \ref{eq_thm}. Also, let $\widehat{\bfh}(s)=\bfA \widehat{\bfc}(s)$ be the estimate of the signal sample $\bfh(s)$ produced by $\ell_{2,1}$-LS minimization of $\eqref{cost_mmv}$,  where $\widehat{\bfc}(s)$ denotes the $s$-th column of the optimal solution $\widehat{\bfC}$ of $\ell_{2,1}$-LS cost function $f(\bfC)$ in \eqref{cost_mmv}. Then, $\widehat{\bfh}(s)$ is given by 
\begin{align}\label{plugin_estim}
\widehat{\bfh}(s)=\bfA \widehat{\Gammam}\bfA^\dagger \Psim(s)^\herm \Big ( \Psim(s)\bfA \widehat{\Gammam}\bfA^\dagger \Psim(s)^\herm  + \varrho \bfI_m \Big )^{-1} \bfx(s),
\end{align}
 where $\bfx(s)=\Psim(s) \bfh(s) + \bfz(s)$ denotes the noisy sketch of $\bfh(s)$ and where $\widehat{\Gammam}=\diag(\widehat{\gammam})$. \hfill$\square$
\end{theorem}
\begin{proof}
Proof in Appendix \ref{mmse_eq_thm_app}.
\end{proof}

\subsection{MMV as a plug-in MMSE estimator}
Using Theorem \ref{eq_thm} and \ref{mmse_eq_thm}, we are now in a position to illustrate the decoupling property of $\ell_{2,1}$-LS algorithm for MMV as claimed in Section \ref{sec:summary} (see also Fig.\,\ref{fig:mmv_dec}). For simplicity of explanation, we will assume a matched scenario where all the signal samples can be written as a sparse linear combination of the columns of the dictionary matrix $\bfA$ as $\bfh(s)=\bfA \bfw(s)$ with some sparse vector of coefficients $\bfw(s) \in \bC^G$ (see also the signal model in \eqref{chan_vec1}). Assuming that the elements of $\bfw(s)=(w_1(s), \dots, w_G(s))^\transp$ are independent with strengths $\gamma^\circ_i=\bE[|w_i(s)|^2]$, one can see that each signal sample $\bfh(s)$ and its noisy sketch  $\bfx(s)=\Psim(s) \bfh(s) + \bfz(s)$ have the following covariance matrices 
\begin{align}\label{cov_mats}
\Sigmam_\bfh=\bfA \Gammam^\circ \bfA^\herm, \ \Sigmam_{\bfx(s)}=\Psim(s)\bfA \Gammam^\circ \bfA^\herm \Psim(s)^\herm + \sigma^2 \bfI_m,\end{align}
 where $\Gammam^\circ=\diag(\gammam^\circ)$ with $\gammam^\circ=(\gamma_1^\circ, \dots, \gamma_G^\circ)^\transp$, and where $\sigma^2$ denotes the noise variance.
%
Using Theorem \ref{eq_thm} and \ref{mmse_eq_thm} and also \eqref{cov_mats}, we can interpret the $\ell_{2,1}$-LS estimator as follows:
\begin{itemize}
\item[{a)}] Obtain an estimate $\widehat{\gammam}$ of the vector $\gammam^\circ$ by optimizing the cost function $g (\gammam)$ in \eqref{cost_gam}, with a regularization $\varrho=\sigma^2$.

\item[{ b)}] Use the estimate $\widehat{\gammam}$ to compute an estimate of the true covariance matrices  in \eqref{cov_mats} as $\Sigmam_\bfh=\bfA \widehat{\Gammam} \bfA^\herm$, and  $\Sigmam_{\bfx(s)}=\Psim(s)\bfA \widehat{\Gammam} \bfA^\herm \Psim(s)^\herm + \varrho \bfI_m$, where $\widehat{\Gammam}=\diag(\widehat{\gammam})$.

\item [{ c)}] Use the estimated covariance matrices in place of the true ones, in a ``plug-in'' MMSE estimator as in \eqref{mmse_intro}, to recover each signal sample $\bfh(s)$  from its corresponding sketches $\bfx(s)$ individually, thus, yielding  \eqref{plugin_estim} in Theorem \ref{mmse_eq_thm}.
\end{itemize}
This  illustrates that $\ell_{2,1}$-LS can be decomposed into a coupled covariance estimation phase (steps a and b) followed by a decoupled signal estimation phase (step c)  (see, e.g., Fig.\,\ref{fig:mmv_dec}). 
This underlying decoupling property also indicates that although regularization based and sample covariance (subspace) based MMV recovery algorithms overviewed in Section \ref{sec:rel_work} seem quite different, they are indeed very related through the covariance estimation phase (steps a and b).

\subsection{Further insights on the decoupling property}
From the decoupling property for $\ell_{2,1}$-LS,  we can gain  several insights on the qualitative performance of $\ell_{2,1}$-LS in the MMV setup as follows:

{\bf i)} The cost function $g(\gammam)$ in \eqref{cost_gam}  has $\tr(\Gammam)=\sum_{i\in [G]} \gamma_i$, i.e.,  the $\ell_1$-norm of  $\gammam$, as the regularizer, which  is known to promote the sparsity of $\gammam$. Thus, we expect that the optimal solution $\widehat{\gammam}$ of $g(\gammam)$ will have only a few significantly large coefficients. Interestingly, this is consistent with the result of Theorem \ref{eq_thm}, where $\widehat{\gamma}_i=\frac{\|\widehat{\bfC}_{i,:}\|_2}{\sqrt{T}}$ and the sparsity of $\widehat{\gammam}$ implies the row-sparsity of the optimal solution of $\ell_{2,1}$-LS in \eqref{cost_mmv}, which in turn follows from the $\ell_{2,1}$-norm regularization.

{\bf ii)} To gain a better understanding of the coupled covariance estimation phase, let us assume a case where  the coefficients have a Gaussian distribution and all the signal coefficients $\{\bfw(s): s\in [T]\}$ are  independent from each other all having the same Gaussian distribution. Then, a straightforward computation shows that one can write the \textit{Maximum Likelihood} (ML) estimate \cite{kay1993fundamentals} of the sparse vector $\gammam^\circ$ (recall that $\gamma^\circ_i=\bE[|w_i(s)|^2]$ denotes the strength of the $i$-th signal coefficient) from the available sketches {$\{\bfx(s): s \in [T]\}$}, which are also independent Gaussian vectors, as the minimization of the following cost function:
\begin{align}
l(\gammam)&=\frac{1}{T} \sum_{s\in [T]} \bfx(s)^\herm \big (\Psim(s)\bfA \Gammam \bfA^\herm \Psim(s)^\herm + \sigma^2 \bfI_m \big )^{-1} \bfx(s) \nonumber \\
& \ \ \ + \log\, \det\big (\Psim(s)\bfA \Gammam \bfA^\herm \Psim(s)^\herm + \sigma^2 \bfI_m \big),\label{cost_ml}
\end{align}
where $\Gammam=\diag(\gammam)$ and where $\det(.)$ denotes the determinant operator. 
From \eqref{cost_ml}, one can see the striking similarity  between the cost function $g(\gammam)$ in \eqref{cost_gam} and ML cost function $l(\gammam)$. One can also note the fundamental difference that the $\log \det$ regularization in \eqref{cost_ml} is replaced by the trace regularization in \eqref{cost_gam} to explicitly promote the sparsity of the vector $\gammam$.

{\bf iii)} The assumptions made for the derivation of  \eqref{cost_ml} reveals that the  $\ell_{2,1}$-norm regularization is in fact unable to capture the statistical correlation among the elements of signal coefficient, i.e., $\{w_i(s): i\in[G]\}$ in $\bfw(s)=(w_1(s), \dots, w_G(s))^\transp$ (spatial correlation) and the correlation among different signal coefficients $\bfw(s)$, $s \in [T]$ (temporal correlation).
In view of the decoupling property derived in Theorem \ref{eq_thm} and \ref{mmse_eq_thm}, in principle, one can improve the performance of  $\ell_{2,1}$-LS by modifying the covariance estimation phase (e.g., as in the cost function \eqref{cost_gam} and \eqref{cost_ml}).
In fact, via this modification, one may obtain a whole variety of MMV schemes consisting of an improved covariance estimation (not necessarily \eqref{cost_gam} or \eqref{cost_ml}), which take into account not only the joint sparsity but also the spatial  signal correlations, followed by the plug-in MMSE estimator as in \eqref{plugin_estim} to estimate signal samples individually.
Notice also that such a family of MMV estimators would  be effective to capture only spatial correlation among the elements of each sparse vector, but not really temporal correlation among different signal samples.

With temporal correlations, one would perhaps need a more involved spatial-temporal covariance estimator to estimate the large-dim $n T \times n T$ joint covariance matrix of all the signal samples $\{\bfh(s): s\in [T]\}$ consisting of $T$ samples, thus,  of $nT$ random variables (each $\bfh(s)$ has dimension $n$), followed by perhaps a  plug-in MMSE estimator applied jointly to all the signal samples. Investigating the connection between such an MMV estimator and other conventional MMV estimators such as $\ell_{2,1}$-LS and deriving a counterpart of decoupling property in those temporally-correlated scenarios is a very interesting problem, but it goes beyond the scope of this paper.

{\bf iv)} One can also see that $\ell_{2,1}$-LS exploits only the second order statistics of the signal but not higher-order ones, thus, it naturally suits Gaussian signals. For non-Gaussian signals, in contrast, one may be able to exploit the higher order signal statistics to obtain better  performances (as is done in Bayesian algorithms such as VAMP discussed in Section \ref{sec:rel_work}).

\subsection{Effect of over-parametrization and dictionary mismatch}\label{sec:dic_mismatch}
As  explained in Section \ref{sec:S-MMV}, when the signal dictionary has infinitely many elements, there may be a mismatch  in $\ell_{2,1}$-LS optimization in \eqref{l1_ls} due to using the quantized dictionary $\bfA$ since the quantization grid $\clG=\{\xi^\circ_1, \dots, \xi^\circ_G\}$ may not contain the true support $\{\xi_1, \dots, \xi_k\}$ (see, e.g., \eqref{chan_vec1}). This  mismatch may also arise even with finite dictionaries when the true sparsifying dictionary is not known precisely \cite{chi2011sensitivity}.

A natural way to reduce this mismatch is to increase the dictionary size by adding additional elements (e.g., increasing the grid size $G$ for a continuous dictionary). However, it is generally believed that enlarging the dictionary size $G$ degrades the recovery performance in a  CS problem such as \eqref{l1_ls} or \eqref{l21_ls}. The underlying reasoning is that by increasing $G$ the columns of the corresponding dictionary matrix $\bfA$ become highly correlated. In addition, sparse estimation under a larger dictionary requires estimating more parameters. For example, for \eqref{l21_ls}, this typically results in more spurious rows in the resulting estimate $\widehat{\bfC}$ and, at first sight, seems to naturally degrade the recovery performance of $\ell_{2,1}$-LS.  

Using the decoupling property of $\ell_{2,1}$-LS, we can now show that although this may, in fact, be true for the recovery of $\widehat{\bfC}$, it does not necessarily hold when the recovery of the original signal samples $\clH=\{\bfh(s): s\in [T]\}$ is of interest. To show this, we need some notation first. 
Let  $\bfA=[\bfa(\xi^\circ_1)\, , \dots, \bfa(\xi^\circ_G)]$ be a finite (or quantized) dictionary consisting of $G$ columns.
We define the cone generated by the columns of $\bfA$ by 
\begin{align}
 \clK_{\bfA}^+=\Big \{\sum_{i \in [G]} p_i \bfa(\xi^\circ_i)\bfa(\xi^\circ_i)^\dagger: p_i \geq 0\Big \}.
\end{align}
Note that $\clK_{\bfA}^+$ is a convex sub-cone of the cone of \textit{positive semi-definite} (PSD) matrices generated by rank-1 matrices $\big\{\bfa(\xi) \bfa(\xi)^\dagger: \xi \in \{\xi_1^\circ, \dots, \xi_G^\circ \} \big \}$. We call a set $\clK \subseteq \bC^{n \times n}$ a convex cone if any positive linear combination of the elements of $\clK$ belongs to $\clK$
 \cite{rockafellar2015convex}. 

A direct inspection in Theorem \ref{mmse_eq_thm} and \eqref{plugin_estim}, reveals that due to the decoupling property of $\ell_{2,1}$-LS,  the dictionary matrix $\bfA$ affects the recovery performance of signal samples  only through the covariance estimate $\widehat{\Sigmam}_{\bfh} =\bfA \widehat{\Gammam} \bfA^\herm$, which belongs to the cone $\clK_{\bfA}^+$ produced by $\bfA$ (note that $\widehat{\Gammam}$ is a diagonal matrix with non-negative diagonal elements). Moreover, assuming without loss of generality that the columns (atoms) of $\bfA$ all have the same $\ell_2$-norm $\zeta:=\|\bfA_{:,i}\|_2$, $i\in [G]$,  using the fact that $\tr(\bfA \Gammam \bfA^\herm)=\zeta^2\,  \tr(\Gammam)$, and by introducing the change of variable $\bfK:=\bfA \Gammam \bfA^\herm$, we can reformulate the optimization of $\Gammam$ in \eqref{cost_gam} more directly in terms of the covariance matrix of the signal $\Sigmam_\bfh$ as follows:
\begin{align}\label{cost_gam2}
\widehat{\Sigmam}_{\bfh}=\argmin _{\bfK \in \clK_{\bfA}^+}  \frac{1}{T} \sum_{s\in [T]} \bfx(s)^\herm & \big(\Psim(s)\bfK \Psim(s)^\herm + \varrho \bfI_m \big)^{-1} \bfx(s) \nonumber\\
&+ \frac{1}{\zeta^2} \tr(\bfK).
\end{align}
This immediately shows that the underlying dictionary can affect  $\ell_{2,1}$-LS only via the covariance estimation phase in  \eqref{cost_gam2}, which depends on the cone $\clK_{\bfA}^+$ produced by $\bfA$. In particular, letting 
$\widetilde{\bfA}$ be a larger dictionary matrix consisting of $\widetilde{G}>G$ columns that contains all $G$ columns of the original dictionary $\bfA$, one can easily see that using the larger dictionary $\widetilde{\bfA}$ does not degrade the performance provided that the additional elements (columns) of $\widetilde{\bfA}$ do not widen/enlarge the convex cone $\clK_{\widetilde{\bfA}}^+$ considerably compared with the original cone $\clK_{\bfA}^+$.

Using the cone produced by the dictionary, we can also obtain a better understanding of the quantization effect in scenarios where the underlying dictionary has infinitely many elements. For such a dictionary $\clA$, we define by  
\begin{align*}
\clK_{\clA}^+:=\Big \{\sum_{\xi \in \Xi'} p_\xi \bfa(\xi) \bfa(\xi)^\herm: p_\xi> 0, \Xi' \subseteq \Xi, |\Xi'|< \infty \Big\},
\end{align*}
the cone generated by $\clA$, where the summation is taken over all possible subsets $\Xi'$ of the labelling set $\Xi$ (see, e.g., \eqref{A_dict}) of finite size (cardinality). In view of \eqref{cost_gam2}, it is not difficult to see that to avoid the effect of dictionary mismatch, one needs to take a quantized dictionary $\bfA=[\bfa(\xi^\circ_1)\,, \dots, \bfa(\xi^\circ_G)]$ over a grid of sufficiently large size $G$ such that the cone $\clK_{\bfA}^+$ approximately exhausts the whole cone $\clK_\clA^+$ (note that $\clK_{\bfA}^+ \subseteq \clK_\clA^+$). In particular, in special cases where the cone $\clK_\clA^+$ has an algebraic representation, namely, it can be represented by finitely many linear constraints over the cone of PSD matrices (note that $\clK_\clA^+$ is a subset of PSD matrices), one can fully eliminate the mismatch by estimating $\Sigmam_\bfh$ via \eqref{cost_gam2} by replacing the constraint set $\clK_{\bfA}^+$ by $\clK_{\clA}^+$ and solving a finite-dim convex optimization problem.

\section{Simulation results}
In this section, we illustrate the validity of the decoupling property for $\ell_{2,1}$-LS via numerical simulations.

\subsection {Basic setup}
For the simulations, we will focus on the sparse channel estimation problem illustrated in Fig.\,\ref{array_fig}.  We assume that the signal dimension is $n=64$ (number of BS antennas) and signal samples $\clH=\{\bfh(s): s\in[T]\}$ are i.i.d. zero-mean Gaussian vectors. We consider a \textit{diffuse} propagation scenario (e.g., reflection from buildings, trees, etc.) where each signal sample $\bfh(s)$ consists of infinitely many scatterers of equal strength in a fixed angular range $\Xi^\circ \subset[-1,1]$. It is straightforward to show that in this case each $\bfh(s)$ is a Gaussian vector with a Toeplitz covariance matrix $\Sigmam_\bfh=\int \bfa(\xi) \bfa(\xi)^\herm \mu(d \xi)$ where  $\bfa(\xi)$ is the array response vector as in \eqref{array_response} and where $\mu(d\xi)$ is a uniform distribution over  $\Xi^\circ$ and denotes the angular power spread function of the signal samples,  which is not a priori known to the estimation algorithm.

\subsection{Effect of dictionary mismatch}\label{sec:sim_dict}
In this section, we investigate the effect of dictionary mismatch as one of the direct consequences of the decoupling property proved for $\ell_{2,1}$-LS (see, Section \ref{sec:dic_mismatch}). We consider a uniform angular power spread function over the angular range $\Xi^\circ=[-0.1,0.1]$. Note that we have intentionally decided to consider a continuous $\mu(d\xi)$ for simulations since, in contrast with discrete AoAs illustrated in Fig.\,\ref{array_fig} and signal model \eqref{chan_vec1}, for such a continuous distribution the underlying dictionary atoms in the sparse representation of signal are  not trivially known.
We consider a window consisting of $T=100$ signal samples and assume that each projection operator $\Psim(s)$, $s\in [T]$, is a random all-zero matrix with only one $1$ in each row  and samples $m=0.5 n$ ($50\%$ sampling) of the components of $\bfh(s)$ (antenna selection). We also assume that the locations of the sampled elements are random in each $\Psim(s)$ and vary i.i.d. across different $s\in [T]$.

We consider three different sparsifying dictionaries: the original continuum dictionary $\clA=\big \{\bfa(\xi): \xi \in [-1,1]\big \}$ as in \eqref{A_dict}, and discrete dictionaries $\clA_n=\big \{\bfa(\frac{2i}{n}-1): i \in [n]\big \}$ and  $\clA_{2n}=\big \{\bfa(\frac{i}{n}-1): i \in [2n]\big \}$ obtained by quantizing $\clA$ over a uniform grid of $\Xi$ of size $n$ and $2n$, respectively. One can check that for the dictionary $\clA$ consisting of complex exponential atoms as in \eqref{array_response}, the dictionary matrices corresponding to $\clA_n$ and $\clA_{2n}$ would be the standard Fourier matrix and the oversampled Fourier matrix with an oversampling factor $2$.  
To obtain the simulation results for $\clA_n$ and $\clA_{2n}$, we can directly solve $\ell_{2,1}$-LS in \eqref{l21_ls} with a regularization parameter $\varrho=\sigma^2$ to estimate $\widehat{\bfC}$ and, thus, $\widehat{\bfh}(s)$ for $s\in [T]$. We can also apply the decoupling property to first estimate the signal covariance matrix by optimizing $g(\gammam)$ in \eqref{cost_gam} and then apply the decoupled plug-in MMSE estimator to estimate the signal samples. In Appendix \eqref{cost_gam_app}, we derive a simple coordinate-wise descend algorithm for optimizing $g(\gammam)$.

For the continuum dictionary $\clA$, however, $\eqref{l21_ls}$ is not directly solvable since, roughly speaking,  it requires using a dictionary matrix with infinitely many columns and estimating a coefficient matrix $\widehat{\bfC}$ with infinitely many rows. Instead, we apply the decoupling property of $\ell_{2,1}$-LS and use  an indirect approach, namely, we solve  \eqref{cost_gam2} to estimate the signal covariance matrix and then apply the plug-in MMSE estimator \eqref{mmse_intro} to estimate the signal samples. To do this, we need to specify the cone generated by the dictionary $\clA$ in \eqref{cost_gam2}. Fortunately, for the dictionary of complex exponentials (array responses in \eqref{array_response}) considered here, it is well-known that the cone generated by the dictionary coincides with the cone of $n \times n$ PSD Hermitian Toeplitz matrices. Since this cone can be parameterized by $n$ complex numbers, the covariance estimation phase in \eqref{cost_gam2} is a finite-dim optimization problem and can be exactly solved. For the simulations, we approximate this optimization by applying the coordinate-wise descend algorithm in  Appendix \eqref{cost_gam_app} but we select the grid size large enough to make sure that the approximation is sufficiently good.

%
%
\begin{figure}[t]
\centering

\begin{tikzpicture}[scale=0.7]

\pgfplotsset{every axis/.append style={
           legend style={font=\large, legend cell align=right},
           label style={font=\large},
           title style={font=\large},
}}

\pgfplotsset{every tick label/.append style={font=\Large}}

\begin{axis}[%
scale only axis,
xmin=0,
xmax=40,
xlabel={Signal-to-Noise Ratio [dB]},
ymode=log,
ymin=0.0005,
ymax=1,
yminorticks=true,
xmajorgrids,
yminorgrids,
ylabel={NMSE},
title={Scaling of NMSE with SNR for $T=100$},
legend style={legend cell align=left, align=left, at={(0.36,0.285)}}
]
\addplot [mark size=3pt, mark=diamond, mark options={solid}, line width=1.5pt]
  table[row sep=crcr]{%
0	0.340294735552627\\
5	0.157902296123059\\
10	0.0658536392258477\\
15	0.0268686401187629\\
20	0.0112373464277831\\
25	0.00491520413422094\\
30	0.0022861055909257\\
35	0.0011462147069155\\
40	0.00062279842630662\\
};
\addlegendentry{MMSE}

\addplot [mark=triangle, color=blue, line width=1.5pt, mark options={solid}]
  table[row sep=crcr]{%
0	0.742748910290386\\
5	0.316271368564708\\
10	0.14698079183245\\
15	0.080731026236735\\
20	0.056073743115848\\
25	0.0473758863837262\\
30	0.0441557426163018\\
35	0.0432266626719735\\
40	0.0441346677133703\\
};
\addlegendentry{$\ell_{2,1}$-LS   $\clA_n$}

\addplot [mark=square, mark size=2pt, color=red, line width=1.5pt]
  table[row sep=crcr]{%
0	0.747417853182191\\
5	0.306811259641631\\
10	0.131906928184556\\
15	0.0616203274194548\\
20	0.0314582035781248\\
25	0.0196118302054447\\
30	0.0131520697803941\\
35	0.010803970883988\\
40	0.00954600250941627\\
};
\addlegendentry{$\ell_{2,1}$-LS   $\clA_{2n}$}

\addplot [mark=*, color=brown, line width=1.5pt, mark options={solid}]
  table[row sep=crcr]{%
0	0.745101731619024\\
5	0.302486819667183\\
10	0.133378630957394\\
15	0.0631008197416917\\
20	0.0328170603574184\\
25	0.0190671152699168\\
30	0.0134849537701604\\
35	0.0102177015284327\\
40	0.00986350398755082\\
};
\addlegendentry{$\ell_{2,1}$-LS   $\clA$}

%
%

\end{axis}
\end{tikzpicture}%

\caption{Comparison of the effect of dictionary mismatch on the performance of $\ell_{2,1}$-LS as an MMV estimator with that of the MMSE estimator for different dictionary sizes. 
}
\label{fig:mmv_mmse}
\end{figure}
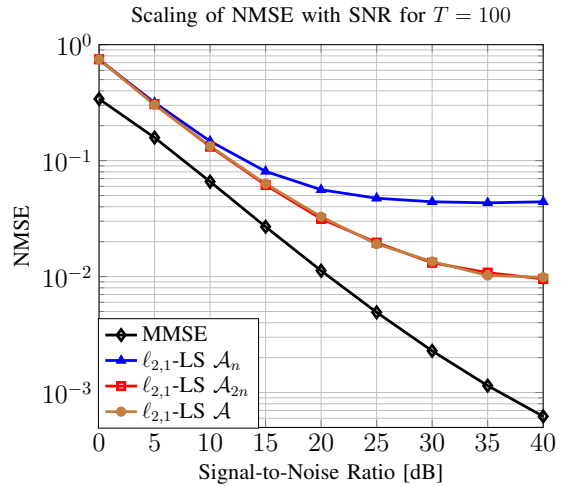

We run the simulations for $100$ i.i.d. realizations of the signal samples to estimate the \textit{Normalized Mean Square Error} (NMSE) of the $\ell_{2,1}$-LS MMV estimator. We compare the NMSE of the MMV estimator with that of the MMSE estimator, where we compute the latter by calculating the Toeplitz covariance matrix $\Sigmam_\bfh$ of the signal samples and applying \eqref{mmse_intro}.
Fig.\,\ref{fig:mmv_mmse} illustrates the simulation results. It is seen that MMV with the dictionary $\clA_{n}$ is highly mismatched and, due to this mismatch, its NMSE gets saturated at high \textit{Signal-to-Noise Ratios} (SNRs). Interestingly, by enlarging the dictionary to $\clA_{2n}$, the NMSE decreases dramatically. It is important to note that  $\ell_{2,1}$-LS over the larger dictionary $\clA_{2n}$ requires estimating twice more number of parameters than that over $\clA_{n}$, but it still yields a better estimation performance.

Another interesting observation one can make from Fig.\,\ref{fig:mmv_mmse} is that by moving form $\clA_{2n}$ to yet a larger (in fact infinite-dim) dictionary, the MMV performance keeps improving but only slightly. This can be qualitatively justified by the fact that the cone of $n\times n$ PSD Toeplitz matrices is well approximated by the cone generated by $\clA_{2n}$. Roughly speaking, this follows from Szeg{\"o} theorem \cite{russell1959u}, which states that the space of $n \times n$ PSD Toeplitz matrices is approximately  diagonalizable  (under an appropriate metric)  with the DFT matrix of dimension $n$ in the asymptotic regime as $n \to \infty$.

%
%
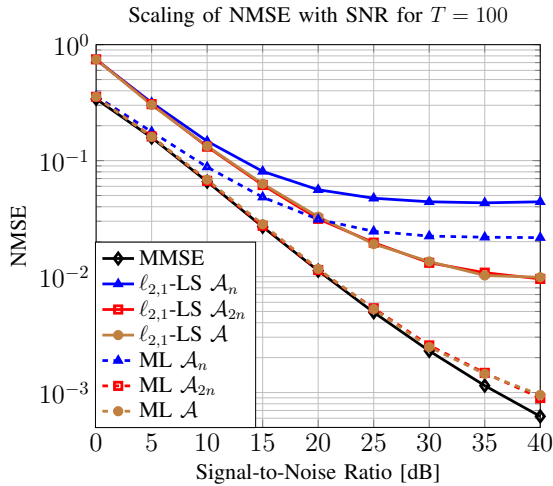
\begin{figure}[t]
\centering

\begin{tikzpicture}[scale=0.7]

\pgfplotsset{every axis/.append style={
           legend style={font=\large, legend cell align=right},
           label style={font=\large},
           title style={font=\large},
}}

\pgfplotsset{every tick label/.append style={font=\Large}}

\begin{axis}[%
scale only axis,
xmin=0,
xmax=40,
xlabel={Signal-to-Noise Ratio [dB]},
ymode=log,
ymin=0.0005,
ymax=1,
yminorticks=true,
xmajorgrids,
yminorgrids,
ylabel={NMSE},
title={Scaling of NMSE with SNR for $T=100$},
legend style={legend cell align=left, align=left, at={(0.36,0.485)}}
]
\addplot [mark size=3pt, mark=diamond, mark options={solid}, line width=1.5pt]
  table[row sep=crcr]{%
0	0.340294735552627\\
5	0.157902296123059\\
10	0.0658536392258477\\
15	0.0268686401187629\\
20	0.0112373464277831\\
25	0.00491520413422094\\
30	0.0022861055909257\\
35	0.0011462147069155\\
40	0.00062279842630662\\
};
\addlegendentry{MMSE}

\addplot [mark=triangle, color=blue, line width=1.5pt, mark options={solid}]
  table[row sep=crcr]{%
0	0.742748910290386\\
5	0.316271368564708\\
10	0.14698079183245\\
15	0.080731026236735\\
20	0.056073743115848\\
25	0.0473758863837262\\
30	0.0441557426163018\\
35	0.0432266626719735\\
40	0.0441346677133703\\
};
\addlegendentry{$\ell_{2,1}$-LS   $\clA_n$}

\addplot [mark=square, mark size=2pt, color=red, line width=1.5pt]
  table[row sep=crcr]{%
0	0.747417853182191\\
5	0.306811259641631\\
10	0.131906928184556\\
15	0.0616203274194548\\
20	0.0314582035781248\\
25	0.0196118302054447\\
30	0.0131520697803941\\
35	0.010803970883988\\
40	0.00954600250941627\\
};
\addlegendentry{$\ell_{2,1}$-LS   $\clA_{2n}$}

\addplot [mark=*, color=brown, line width=1.5pt, mark options={solid}]
  table[row sep=crcr]{%
0	0.745101731619024\\
5	0.302486819667183\\
10	0.133378630957394\\
15	0.0631008197416917\\
20	0.0328170603574184\\
25	0.0190671152699168\\
30	0.0134849537701604\\
35	0.0102177015284327\\
40	0.00986350398755082\\
};
\addlegendentry{$\ell_{2,1}$-LS   $\clA$}

\addplot [mark=triangle, color=blue, dashed, line width=1.5pt, mark options={solid}]
  table[row sep=crcr]{%
0	0.36014278576642\\
5	0.177305654218243\\
10	0.0880233466369256\\
15	0.0484025536286685\\
20	0.0311246978040278\\
25	0.0245598275156336\\
30	0.0223351744500674\\
35	0.0218392412326365\\
40	0.0216625189837455\\
};
\addlegendentry{ML $\clA_n$}

\addplot [mark=square, mark size=2pt, color=red, dashed, line width=1.5pt, mark options={solid}]
  table[row sep=crcr]{%
0	0.355016975566755\\
5	0.161656462149694\\
10	0.0669284746724812\\
15	0.0274192239937346\\
20	0.0113776865010595\\
25	0.00532863076982947\\
30	0.00254757801211201\\
35	0.0014774716122582\\
40	0.00089679075224339\\
};
\addlegendentry{ML $\clA_{2n}$}

\addplot [mark=*, color=brown, dashed, line width=1.5pt, mark options={solid}]
  table[row sep=crcr]{%
0	0.355600710511665\\
5	0.160161734879902\\
10	0.0683938667653514\\
15	0.0282698899914172\\
20	0.0116430744836451\\
25	0.00524697347151797\\
30	0.0024695408499463\\
35	0.00144904584248481\\
40	0.000953444380748618\\
};
\addlegendentry{ML $\clA$}

\end{axis}
\end{tikzpicture}%

\caption{Illustration of improvement of the performance of $\ell_{2,1}$-LS via  a better covariance estimation phase for different dictionaries and SNRs. 
}
\label{fig:mmv_ml_mmse}
\end{figure}

\subsection{Improving the $\ell_{2,1}$-LS via a better covariance estimation} \label{sec:sim_improve}
As explained so far in this paper, one can see $\ell_{2,1}$-LS MMV signal estimation as a plug-in MMSE estimator, which performs a coupled covariance estimation via optimizing \eqref{cost_gam} or equivalently \eqref{cost_gam2} and uses the resulting covariance estimate in a decoupled signal sample estimation. A natural consequence of this decoupling property is that one can potentially improve the performance of $\ell_{2,1}$-LS via improving the covariance estimation phase. In this section, our goal is to investigate this further. We consider a \textit{completely new} MMV estimator that has the same decoupled structure as $\ell_{2,1}$-LS with the difference that it estimates the signal covariance matrix using ML technique by optimizing $l(\gammam)$ in \eqref{cost_ml} rather than optimizing $g(\gammam)$ in \eqref{cost_gam} as applied in $\ell_{2,1}$-LS. Note that in contrast with $g(\gammam)$, which is  a convex function of $\gammam$ and can be exactly optimized, $l(\gammam)$ is generally a non-convex function of $\gammam$. In Appendix \ref{cost_ml_app}, we derive a simple coordinate-wise descend algorithm for minimzing $l(\gammam)$. This algorithm converges to local minimizer of $l(\gammam)$, which is not guaranteed to be the global minimizer. However, our simulations show that this algorithm always yields a reasonably good solution.

For simulations, we use a uniform $\mu(d\xi)$ over $\Xi^\circ=[-0.1, 0.1]$ and set $T=100$ as in Section \ref{sec:sim_dict}. Fig.\,\ref{fig:mmv_ml_mmse} illustrates the simulation results. It is seen that modifying the covariance estimation phase in the decoupled representation of $\ell_{2,1}$-LS indeed improves its performance dramatically. In fact, for sufficiently large dictionary sizes, our newly-proposed MMV estimator achieves a performance very close to that of  MMSE estimator even for sample size of $T=100$. 

This clearly shows that, by using the decoupling property proved in this paper and by simply modifying and improving the coupled covariance estimation phase, we may obtain a variety of MMV estimation algorithms with performances even much better than that of $\ell_{2,1}$-LS.



\section{Conclusions and further remarks}
In this paper, we proved a decoupling property for $\ell_{2,1}$-LS, which is a well-known and widely-adopted  algorithm for joint signal estimation in MMV setups.
Specifically, we proved that $\ell_{2,1}$-LS is equivalent to the following two-step procedure:
i)\,use the whole sketches obtained from all signal samples to estimate their covariance matrix (coupled phase), ii)\,plug in the resulting covariance estimate as the true covariance matrix into the MMSE estimator to reconstruct each signal sample from its corresponding sketch individually (decoupled phase).
Using the decoupling property, we could provide several insights on the performance of the $\ell_{2,1}$-LS as an MMV estimator. 
For example, we explained that, due to its  linear form as the plug-in MMSE estimator,  $\ell_{2,1}$-LS is unable to take advantage of the higher-oder  statistics of the signal and  is far from optimal for sparse non-Gaussian  processes.  Moreover, using this property, we were able to quantify the effect of dictionary mismatch on the performance of $\ell_{2,1}$-LS when the underlying sparsifying dictionary of the process is not precisely known. For this specific case, we provided numerical simulations to illustrate the validity of the  decoupling property. Also, motivated by the decoupling property of $\ell_{2,1}$-LS, we designed a totally new MMV algorithm by modifying the coupled covariance estimation phase and illustrated that it yields a better performance than $\ell_{2,1}$-LS.


%
%


\appendices
\section{Proofs}

\subsection{Proof of Theorem \ref{eq_thm}}\label{eq_thm_app}
The proof follows by extending Theorem 1 in \cite{steffens2016compact}.
The key observation is that for a vector $\bfu \in \bC^T$, the $\ell_2$-norm $\|\bfu\|_2$ can be written as the output of the following optimization 
\begin{align}\label{eq:norm_rep}
\|\bfu\|_2=\min_{\bfv \in \bC^T,\,d \in \bC:\, d\bfv =\bfu} \frac{1}{2}(\|\bfv\|_2^2 + |d|^2).
\end{align}
In particular, $\|\bfu\|_2=|d^*|^2$, where $d^*$ is the optimal solution of \eqref{eq:norm_rep}. Applying this to the rows of a $G\times T$ matrix $\bfC$, we can write the $\ell_{2,1}$-norm of $\bfC$ as follows
\begin{align}\label{l21_rep}
\|\bfC\|_{2,1}=\min_{\bfV\in \bC^{G\times  T},\, \bfD\in \bD:\, \bfD \bfV=\bfC} \frac{1}{2}(\|\bfV\|_{\sfF}^2 + \|\bfD\|_{\sfF}^2),
\end{align}
where $\|.\|_\sfF$ denotes the Frobenius norm, where $\bD$ denotes the space of $G\times G$ diagonal matrices with diagonal elements in $\bC$, and where $\bfD=\diag(d_1, \dots, d_G)\in \bD$. Furthermore, we have the identity
\begin{align}\label{key_id}
\|\bfC_{i,:}\|_2=|d_i^*|^2,
\end{align}
 with $d^*_i$ denoting the $i$-th diagonal element of the optimal solution $\bfD^*=\diag(d^*_1, \dots, d_G^*)$ in \eqref{l21_rep}. 

Using this result and  by replacing $\|\bfC\|_{2,1}$ in \eqref{cost_mmv} with \eqref{l21_rep}, we can transform the optimization problem $\widehat{\bfC}=\argmin_{\bfC\in \bC^{G\times T}} f(\bfC)$ in \eqref{cost_mmv} into 
\begin{align}\label{trans_eq}
(\bfV^*, \bfD^*)=\argmin_{\bfV\in \bC^{G\times  T},\,\bfD\in \bD} &\frac{1}{\varrho \sqrt{T}} \sum_{s=1}^T \|\Psim(s)\bfA (\bfD \bfV)_{:,s} - \bfx(s)\|_{2}^2 \nonumber\\
&+ \|\bfV\|_{\sfF}^2 + \|\bfD\|_{\sfF}^2.
\end{align}
For a fixed $\bfD$, the minimizing $\bfV$ as a function of $\bfD$ can be obtained via a least-squares optimization, where after replacing the solution in \eqref{trans_eq} and applying additional simplification, and using $\|\bfD\|_\sfF^2=\tr(\bfD \bfD^\herm)$,  we obtain the following optimization in terms of the remaining  variable $\bfD$
%
\begin{align}\label{l21_rep2}
\bfD^*=\argmin_{\bfD \in \bD} \frac{1}{T} \hspace{-1mm} \sum_{s\in [T]} \hspace{-1mm} \bfx(s)^\herm &\Big ( \Psim(s) \bfA \frac{\bfD \bfD^\herm}{\sqrt{T}} \bfA^\herm \Psim(s)^\herm + \varrho \bfI_m\Big)^{-1} \hspace{-2mm} \bfx(s)\nonumber\\
&+\tr(\frac{\bfD \bfD^\herm}{\sqrt{T}}).
\end{align}
This optimization can be reparameterized with $\Gammam=\frac{\bfD \bfD^\herm}{\sqrt{T}}=\diag(\frac{|d_1|^2}{\sqrt{T}}, \dots, \frac{|d_G|^2}{\sqrt{T}}) \in \bD_+$, where $\bD_+$ denotes the space of all $G\times G$ diagonal matrices with non-negative diagonal elements. With this reparameterization, we obtain 
\begin{align}\label{l21_rep3}
\widehat{\Gammam}=\argmin_{\Gammam \in \bD_+}   
\frac{1}{T}\sum_{s=1}^T \bfx(s)^\herm&  \Big ( \Psim(s) \bfA \Gammam \bfA^\herm \Psim(s)^\herm
+ \varrho \bfI_m\Big)^{-1} \bfx(s) \nonumber\\ 
&+ \tr(\Gammam).
\end{align}
It is immediately seen that the optimal solution $\widehat{\Gammam}$ (a diagonal matrix with non-negative diagonal elements $\widehat{\gammam} \in \bR_+^G$) coincides with the optimal solution of the cost function $g(\gammam)$ (with $\Gammam=\diag(\gammam)$) in \eqref{cost_gam}. Moreover, using the reparametrization $\Gammam=\tr(\frac{\bfD \bfD^\herm}{\sqrt{T}})$ in \eqref{l21_rep2}, the fact that the optimal solution $\widehat{\bfC}$ of $f(\bfC)$ in \eqref{cost_mmv}  is equivalently given by $\bfD^* \bfV^*$ in terms of the optimal solution $( \bfV^*, \bfD^*)$ of \eqref{trans_eq}, and also  the identity \eqref{key_id}, we have  
\begin{align}
\widehat{\gamma}_i=\frac{|d_i^*|^2}{\sqrt{T}}=\frac{\|\widehat{\bfC}_{i,:}\|_2}{\sqrt{T}},
\end{align}
where $\widehat{\gamma}_i$ denotes the $i$-th component of $\widehat{\gammam}$. This shows the underlying connection between the optimal solution of $f(\bfC)$ in \eqref{cost_mmv} and that of $g(\gammam)$ in \eqref{cost_gam} as stated in the theorem. This completes the proof.

\subsection{Proof of Theorem \ref{mmse_eq_thm}}\label{mmse_eq_thm_app}
Let $\widehat{\bfC}=\argmin_{\bfC \in \bC^{G \times T}} f(\bfC)$ be the optimal solution of \eqref{cost_mmv} and let $s \in [T]$ be a fixed number. From the optimality of $\widehat{\bfC}$, it results that $\widehat{\bfc}(s):=\widehat{\bfC}_{:,s}$ (the $s$-th column of $\widehat{\bfC}$) is the optimal solution of the function $f_s: \bC^G \to \bC$ defined by
\begin{align}\label{g_w}
f_s(\bfc(s)):=\frac{1}{2}\|\Psim(s) \bfA \bfc(s) - \bfx(s)\|_2^2 + \varrho \sqrt{T} \Big \|[\widehat{\bfC}_{ \sim s}, \bfc(s)]\Big \|_{2,1},
\end{align}
where $\widehat{\bfC}_{ \sim s}$ is the $G\times (T-1)$ matrix obtained by removing the $s$-th column of $\widehat{\bfC}$ and where $[\widehat{\bfC}_{ \sim s}, \bfc(s)]$ denotes the $G\times T$ matrix obtained by adjoining $\bfc(s)$ as the $s$-th  column to $\widehat{\bfC}_{ \sim s}$.  Note that if $\widehat{\bfC}_{i,:}=0$, namely, the $i$-th row of the optimal solution $\widehat{\bfC}$ is an all-zero vector,  then the optimal solution $\widehat{\bfc}(s)=(\widehat{c}_1(s), \dots, \widehat{c}_G(s))^\transp$ of \eqref{g_w} should satisfy $\widehat{c}_i(s)=0$. So we focus only on the nonzero rows of  $\widehat{\bfC}$. Let $i\in [G]$ be such a row. Then, taking the derivative of the cost function $f_s(\bfc(s))$ in \eqref{g_w} with respect to $c_i(s)$, we obtain 
\begin{align}\label{eq_w_opt}
\bfA_{:,i}^\herm \Psim(s)^\herm  \big(\Psim(s) \bfA  \widehat{\bfc}(s) -\bfx(s)\big)+ \varrho \sqrt{T} \frac{\widehat{c}_i(s)}{\|\widehat{\bfC}_{i,:}\|_2}={0}.
\end{align}
By introducing the diagonal matrix $\bfL$ with the diagonal elements $\bfL_{i,i}=\frac{\varrho \sqrt{T}}{\|\widehat{\bfC}_{i,:}\|_2}$, we can write \eqref{eq_w_opt} for the non-zero elements more compactly as 
\begin{align}\label{ll_dummy}
\bfA^\herm \Psim(s)^\herm \big (\Psim(s) \bfA \widehat{\bfc}(s) -\bfx(s)\big) +\bfL \widehat{\bfc}(s)={\bf 0}.
\end{align}
 Note that if the $i$-th row of $\widehat{\bfC}$ is zero, then $\bfL_{i,i}=\infty$, which forces the $i$-th component of $\widehat{\bfc}(s)$ to be zero. Thus, with some abuse of notation, \eqref{ll_dummy} also holds for the zero elements. Thus,
\begin{align}
\widehat{\bfc}(s)&=\big(\bfA^\herm \Psim(s)^\herm  \Psim(s)  \bfA + \bfL \big )^{-1}  \bfA^\herm \Psim(s)^\herm \bfx(s)\nonumber\\
&\stackrel{(a)}{=}\bfL^{-1} \bfA^\herm \Psim(s)^\herm \big ( \Psim(s) \bfA \bfL^{-1} \bfA^\herm\Psim(s)^\herm   + \bfI_m \big)^{-1} \bfx(s)\nonumber\\
&\stackrel{(b)}{=}\widehat{\Gammam} \bfA^\herm \Psim(s)^\herm \big ( \Psim(s)\bfA \widehat{\Gammam}\bfA^\dagger \Psim(s)^\herm  + \varrho \bfI_m \big )^{-1} \bfx(s),\nonumber
\end{align}
where in $(a)$ we applied the matrix inversion lemma, and where in $(b)$ we used Theorem \ref{eq_thm} and the fact that the diagonal element of $\bfL^{-1}$ are given by $(\bfL^{-1})_{i,i}=\frac{\|\widehat{\bfC}_{i,:}\|_2}{\varrho \sqrt{T}}$ and  replaced $\bfL^{-1}=\varrho \widehat{\Gammam}$. 
%
This yields the signal estimate $\widehat{\bfh}(s)=\bfA \widehat{\bfc}(s)$  as 
\begin{align*}
\widehat{\bfh}(s)=\bfA \widehat{\Gammam}\bfA^\dagger \Psim(s)^\herm \Big( \Psim(s)\bfA \widehat{\Gammam}\bfA^\dagger \Psim(s)^\herm  + \varrho \bfI_m\Big)^{-1} \bfx(s).
\end{align*}
This completes the proof.

\subsection{Coordinate-wise steepest descend for $g(\gammam)$}\label{cost_gam_app}
We consider the cost function $g(\gammam)$ as in \eqref{cost_gam}
\begin{align*}
g(\gammam)=\frac{1}{T}\sum_{s\in [T]}\hspace{-2mm}  \bfx(s)^\herm \big (\Psim(s)\bfA \Gammam \bfA^\herm \Psim(s)^\herm + \varrho \bfI_m \big )^{-1} \bfx(s) + \tr(\Gammam).
\end{align*}
We define $g_k(d)=g(\gammam + d \bfe_k)$ as the function $g(\gammam)$ restricted to its $k$-th components $\gamma_k$ where $\bfe_k$ denotes the $G$-dim canonical vector with $1$ at its $k$-th element and zero elsewhere. We also define $\Sigmam_s( \gammam)=\Psim(s) \bfA \Gammam \bfA^\herm \Psim(s)^\herm + \varrho \bfI_m$. Using $\Sigmam_s(\gammam+d \bfe_k)=\Sigmam_s(\gammam) + d \bfa_s(k) \bfa_s(k)^\herm$, where $\bfa_s(k)$ denotes the $k$-th column of $\Psim(s) \bfA$, and using the rank-1 update identity for the inverse of a matrix \cite{sherman1950adjustment}, we have 
\begin{align*}
\Sigmam_s(\gammam+d \bfe_k)^{-1}=\Sigmam_s(\gammam)^{-1} - \frac{d \Sigmam_s(\gammam)^{-1} \bfa_s(k) \bfa_s(k)^\herm  \Sigmam_s(\gammam)^{-1} }{1+ d \bfa_s(k) ^\herm \Sigmam_s(\gammam)^{-1}\bfa_s(k) }.
\end{align*}
Applying this identity, we can write $g_k(d)$ as follows:
\begin{align}
g_k(d)= g(\gammam) + d - \frac{1}{T}\sum_{s\in[T]} \frac{d \|\bfa_s(k)^\herm  \Sigmam_s(\gammam)^{-1} \bfx(s)\|^2}{1+ d \bfa_s(k) ^\herm \Sigmam_s(\gammam)^{-1}\bfa_s(k) },
\end{align}
where $g(\gammam)$ does not depend on $d$. Note that $g_k(d)$ is a convex function of $d$ with the derivative
\begin{align}
g_k'(d)= 1- \frac{1}{T}\sum_{s\in[T]} \frac{\|\bfa_s(k)^\herm  \Sigmam_s(\gammam)^{-1} \bfx(s)\|^2}{\big (1+ d \bfa_s(k) ^\herm \Sigmam_s(\gammam)^{-1}\bfa_s(k)\big )^2 }.
\end{align}
It is not difficult to see that since $\Sigmam_s(\gammam)$, $s \in [T]$, are all PSD matrices, $g'_k(d)$ is an increasing function of $d$ in the region $d \in (d_{\min}, \infty)$ with $d_{\min}=\max\{d_{\min,1}, d_{\min,2}\}$,  where $d_{\min,1}=- \frac{1}{\max_{s\in [T]} \bfa_s(k) ^\herm \Sigmam_s(\gammam)^{-1}\bfa_s(k)}$ and where $d_{\min,2}=-\gamma_k$ comes from the fact that $\gamma_k + d$ should be non-negative.

As a result, the minimizer $d^*$ of $g_k(d)$ is given by the unique solution of $g_k'(d^*)=0$ when it has a solution in $d \in (d_{\min},\infty)$ or by $d^*=d_{\min}$ when it does not have any solution. Moreover, in the former case, due to the convexity of $g_k(d)$ the derivative $g_k'(d)$ is an increasing function of $d$ and its root can be easily found by applying a bisection algorithm in the region $d \in (d_{\min},\infty)$.

\subsection{Coordinate-wise steepest descend for $l(\gammam)$}\label{cost_ml_app}
We consider the cost function $l(\gammam)$ as in \eqref{cost_ml}
\begin{align}
l(\gammam)&= \frac{1}{T} \sum_{s\in [T]} \bfx(s)^\herm \big (\Psim(s)\bfA \Gammam \bfA^\herm \Psim(s)^\herm + \sigma^2 \bfI_m \big )^{-1} \bfx(s) \nonumber \\
& \ \ \ + \log\, \det\big (\Psim(s)\bfA \Gammam \bfA^\herm \Psim(s)^\herm + \sigma^2 \bfI_m \big),
\end{align}
We define $l_k(d)=g(\gammam + d \bfe_k)$ as the function $g(\gammam)$ restricted to its $k$-th components $\gamma_k$ where $\bfe_k$ denotes the $G$-dim canonical vector with $1$ at its $k$-th element and zero elsewhere. We also define $\Sigmam_s( \gammam)=\Psim(s) \bfA \Gammam \bfA^\herm \Psim(s)^\herm + \varrho \bfI_m$. Using $\Sigmam_s(\gammam+d \bfe_k)=\Sigmam_s(\gammam) + d \bfa_s(k) \bfa_s(k)^\herm$, where $\bfa_s(k)$ denotes the $k$-th column of $\Psim(s) \bfA$, and using the rank-1 update identity for the inverse and the determinant  of a matrix \cite{sherman1950adjustment}, we have 
\begin{align*}
&\Sigmam_s(\gammam+d \bfe_k)^{-1}\hspace{-1mm} =\hspace{-1mm} \Sigmam_s(\gammam)^{-1} - \frac{d \Sigmam_s(\gammam)^{-1} \bfa_s(k) \bfa_s(k)^\herm  \Sigmam_s(\gammam)^{-1} }{1+ d \bfa_s(k) ^\herm \Sigmam_s(\gammam)^{-1}\bfa_s(k) },\\
&\det\big (\Sigmam_s(\gammam+d \bfe_k)\big )= \det\big (\Sigmam_s(\gammam)\big ) \\
& \hspace{35mm} \times \big(1+ d \bfa_s(k)^\herm \Sigmam_s(\gammam)^{-1} \Sigmam_s(\gammam) \big ).
\end{align*}
Applying this identity, we can write $g_k(d)$ as follows:
\begin{align}
g_k(d)&= g(\gammam) - \sum_{s\in[T]} \frac{d \|\bfa_s(k)^\herm  \Sigmam_s(\gammam)^{-1} \bfx(s)\|^2}{1+ d \bfa_s(k) ^\herm \Sigmam_s(\gammam)^{-1}\bfa_s(k) }\nonumber\\
&\ \ \ + \log(1+ d \bfa_s(k) ^\herm \Sigmam_s(\gammam)^{-1}\bfa_s(k)).
\end{align}
where $g(\gammam)$ does not depend on $d$. Note that $g_k(d)$ is generally a non-convex function of $d$ with the derivative
\begin{align}
g_k'(d)&= - \sum_{s\in[T]} \frac{\|\bfa_s(k)^\herm  \Sigmam_s(\gammam)^{-1} \bfx(s)\|^2}{(1+ d \bfa_s(k) ^\herm \Sigmam_s(\gammam)^{-1}\bfa_s(k))^2 }\nonumber \\
& \ \ \ \ + \frac{\bfa_s(k) ^\herm \Sigmam_s(\gammam)^{-1}\bfa_s(k)}{1+ d \bfa_s(k) ^\herm \Sigmam_s(\gammam)^{-1}\bfa_s(k)}.
\end{align}
Also, we have that $g'_k(d)>0$ for sufficiently large $d$ and that $g'_k(d_{\min}^+)=-\infty$ for $d_{\min}=- \frac{1}{\max_{s\in [T]} \bfa_s(k) ^\herm \Sigmam_s(\gammam)^{-1}\bfa_s(k)}$.  This implies that $g'_k(d)=0$ has at least one solution in the domain $d \in (d_{\min}, \infty)$. However, in general, it might have more than one solution. Here, we always use the largest solution and find it using the bisection method. Note that due to the structure of $g'_k(d)$ and especially the fact that $g'_k(d)>0$ for sufficiently large $d$, this solution would correspond to a local minimizer of $g_k(d)$. Denoting this solution by $d_0$ and using the fact that $d_0+\gamma_k$ should be non-negative, we set the optimal minimizer of $g_k(d)$ to $d^*=\max\{d_0, -\gamma_k \}$.

\balance
{\small
\bibliographystyle{IEEEtran}
\bibliography{references2}
}

\end{document}